\documentclass[10pt,conference,letterpaper]{IEEEtran}
\usepackage{enumerate}
\usepackage{xcolor}
\usepackage{cite}
\usepackage[cmex10]{amsmath}
\usepackage{amssymb, amsmath}
\usepackage{amsthm}
\usepackage{graphicx}
\graphicspath{{./}}
\usepackage[caption=false,font=footnotesize]{subfig}
\usepackage{bm}
\usepackage{pifont} 
\usepackage{mathrsfs}

\usepackage{algorithm,algorithmic}

\usepackage{physics}

\newcommand{\nop}[1]{}

\newtheorem{definition}{Definition}
\newtheorem{example}{Example}
\newtheorem{theorem}{Theorem}
\newtheorem{corollary}{Corollary}

\newtheorem{proposition}{Proposition}
\newtheorem{remark}{Remark}

\begin{document}

\title{Stochastic Dependence in Wireless Channel Capacity: A Hidden Resource}

\author{
Fengyou Sun and Yuming Jiang\\
NTNU, Norwegian University of Science and Technology, Trondheim, Norway\\
sunfengyou@gmail.com, ymjiang@ieee.org
}

\maketitle

\begin{abstract}
This paper dedicates to exploring and exploiting the hidden resource in wireless channel.
We discover that the stochastic dependence in wireless channel capacity is a hidden resource, specifically, if the wireless channel capacity bears negative dependence, the wireless channel attains a better performance with a smaller capacity. We find that the dependence in wireless channel is determined by both uncontrollable and controllable parameters in the wireless system, by inducing negative dependence through the controllable parameters, we achieve dependence control.
We model the wireless channel capacity as a Markov additive process, i.e., an additive process defined on a Markov process, and we use copula to represent the dependence structure of the underlying Markov process. Based on a priori information of the temporal dependence of the uncontrollable parameters and the spatial dependence between the uncontrollable and controllable parameters, we construct a sequence of temporal copulas of the Markov process, given the initial distributions, we obtain a sequence of transition matrices of the controllable parameters satisfying the expected dependence properties of the wireless channel capacity.
The goal of this paper is to show the improvement of wireless channel performance from transforming the dependence structures of the capacity.
\end{abstract}


\IEEEpeerreviewmaketitle

\section{Introduction}

Wireless communication has been around for over a hundred years, starting in 1896 with Marconi's successful demonstration of wireless telegraphy and transmission of the first wireless signals across the Atlantic in 1901 \cite{niehenke2014wireless}.
For cellular systems, the first generation is deployed in around 1980s, i.e., 1G, then 2G in 1990s, 3G in 2000s, and 4G in 2010s \cite{niehenke2014wireless}.
It has become a pattern that a new generation of wireless system is deployed every a decade and the theme of each generation is to increase the capacity and spectral efficiency of wireless channels. 
The trend is driven by the explosion of wireless traffic that is a rough reflection of people's demand on wireless communication, and the paradox of supply and demand \cite{hecht2016bandwidth} is kept relieving generation by generation through exploiting the physical resources, i.e., power, diversity, and degree of freedom \cite{tse2005fundamentals}.

Considering trillions of devices to be connected to the wireless network, high capacity demand, and stringent latency requirement in the coming 5G \cite{andrews2014will}, it's imperative to rethink the wireless channel resources.
We present a perspective on the challenge by asking and answering a question in this paper.
\begin{itemize}
\item
What's the hidden resource and how to use it?

We discover that the stochastic dependence in wireless channel capacity is the hidden resource and we find a way to achieve dependence control.
We classify the dependence into three categories, i.e., positive dependence, independence, and negative dependence, and we identify the dependence as a resource, because when the wireless channel capacity bears negative dependence relative to positive dependence, the wireless channel attains a better performance with a smaller average capacity with respect to independence.
Being exploitable is a further requirement of a benign resource, specifically, the dependence in wireless channel is induced both by the uncontrollable parameters, e.g., fading, and by the controllable parameters, e.g., power, 
we model the dependence caused by these random parameters with multivariate copula, and we propose to use the controllable dimensions to transform the dependence in whole capacity process.
While the multivariate dependence nature of wireless channel capacity is complex, the diversity of dependence structures in different dimensions provides an opportunity to achieve dependence control and improve channel performance.
\end{itemize}

The copula property of Markov process is investigated in \cite{darsow1992copulas} and extended to high order case in \cite{ibragimov2009copula} and multivariate case in \cite{overbeck2015multivariate}. No-Granger causality is a concept in econometrics and its relation with Markov process is investigated in \cite{cherubini2011copula}.
We model the random parameters in wireless channel capacity as a multivariate Markov process, the Markov family copula in \cite{darsow1992copulas,overbeck2015multivariate} are used not only as a mechanism for dependence modelling, i.e., the copula is an expression of dependence structure, but also as a tool for dependence controlling, i.e., the copula function provides a solution to the controllable parameter configuration.
We apply the no-Granger causality to model the relationship between the controllable and uncontrollable parameters, and the sufficient and necessary condition for Markov process is extended from the bivariate case in \cite{cherubini2011copula} to the multivariate case in this paper.

The remainder of this paper is structured as follows.
Sec. \ref{model} dedicates to modelling dependence in wireless channel capacity.
First, basic capacity concepts are introduced, including instantaneous capacity, cumulative capacity, and transient capacity; 
second, the capacity is modeled as a Markov additive process and the Markov property is expressed by copula; 
third, the distribution function of the Markov additive capacity is investigated and lower and upper bounds are derived.
Sec. \ref{control} proposes an approach to control dependence.
First, performance measures of the wireless channel are derived, namely delay, backlog, and delay-constrained capacity; second, the dependence is classified into three types, i.e., positive dependence, independence, and negative dependence, according to the performance measure, it's proved that the negative dependence is good to the channel performance, while the positive dependence does the opposite thing; last, the cause of dependence in capacity is distinguished as uncontrollable parameters and controllable parameters, it's elaborated how to use the controllable parameters to induce negative dependence to the capacity, and an example of using power as a controllable parameter is illustrated.
Finally, the paper is concluded in Sec. \ref{conclusion}.

\section{Dependence Model}\label{model}

\subsection{Channel Capacity}

Consider a flat fading channel with input $x(t)$, output $y(t)$, fading process $h(t)$, and additive white Gaussian noise process $n(t)\sim\mathcal{C}\mathcal{N}(0,N_0)$,
the complex baseband representation is expressed as \cite{goldsmith2005wireless,tse2005fundamentals}
\begin{equation}
y(t) = h(t)x(t) + n(t),
\end{equation}
conditional on a realization of $h(t)$, the mutual information 
is expressed as \cite{goldsmith2005wireless} 
\begin{IEEEeqnarray}{rCl}\label{eq-2}
I(X;Y|h(t)) = \sum\limits_{x\in\mathcal{X},y\in\mathcal{Y}} P(x,y|h_t)\log_2\frac{P(x,y|h_t)}{P(x|h_t)P(y|h_t)}, \IEEEeqnarraynumspace
\end{IEEEeqnarray}
where $\mathcal{X}$ and $\mathcal{Y}$ are respectively the input and output alphabets of the channel. 
For multiple input and multiple output channel, the generalized formula is available in \cite{telatar1999capacity,foschini1998limits}.

The maximum mutual information over input distribution at $t$, denoted as $C(t)$, is defined as { instantaneous capacity} \cite{costa2010multiple}:
\begin{equation}\label{eq-1}
C(t) = \max\limits_{P(x)}I(X;Y|h(t)).
\end{equation}
The sum of instantaneous capacity in discrete time $(s,t]$ or the integral of instantaneous capacity in continuous time $[s,t)$, denoted as $S(s,t)$, is defined as { cumulative capacity}:
\begin{eqnarray}
S(s,t) = \sum\limits_{i=s+1}^{t}{C(i)}\ 
\left(\text{or\ } S(s,t) = \int_{s}^{t}C(\tau)d\tau \right). 
\end{eqnarray}
Denote $S(t)\equiv S(0,t)$.
The time average of the cumulative capacity through $[0,t)$ is defined as { transient capacity}:
\begin{equation}
\overline{C}(t) = \frac{S(t)}{t}. 
\end{equation}

\begin{example}
For a single input single output channel, 
if the channel side information is only known at the receiver, the instantaneous capacity is expressed as \cite{tse2005fundamentals}
\begin{equation}\label{def-ic}
C(t) = W\log_{2}\left(1+\gamma|h(t)|^{2}\right),
\end{equation}
where $|h(t)|$ denotes the envelope of $h(t)$, $\gamma = {P}/{N_{0}W}$ denotes the average received SNR per complex degree of freedom, $P$ denotes the average transmission power per complex symbol, $N_{0}/2$ denotes the power spectral density of AWGN, and $W$ denotes the channel bandwidth. 
\end{example}

\subsection{Markov Dependence}

Let $(\Omega,\bm{\mathscr{F}},(\bm{\mathscr{F}}_t)_{t\in{N}},P)$ be a filtered probability space and $(\bm{X}_t)_{t\in{N}}$ be an adapted stochastic process. $\bm{X}$ is a Markov process if and only if
\begin{IEEEeqnarray}{rCl}
P\left( \bm{X}_t\le{x}|\bm{X}_{t-1}, \bm{X}_{t-2}, \ldots, \bm{X}_0 \right) = P\left( \bm{X}_t\le{x}|\bm{X}_{t-1} \right). \IEEEeqnarraynumspace
\end{IEEEeqnarray}
The Markov property is solely a dependence property that can be modeled exclusively in terms of copulas \cite{darsow1992copulas,overbeck2015multivariate}.
(Copula is introduced in Appendix \ref{copula_introduction}.)

The $n$-dimensional process $\mathbf{X}$ is a Markov process, if and only if, for all $t_1< t_2<\ldots < t_p$, the copula $C_{t_1,\ldots,t_p}$ of $(\mathbf{X}_{t_1},\ldots,\mathbf{X}_{t_p})$ satisfies \cite{overbeck2015multivariate}
\begin{IEEEeqnarray}{rCl}
C_{t_1,\ldots,t_p} = C_{t_1,t_2}\stackrel{C_{t_2}(.)}{\star}C_{t_2,t_3}\stackrel{C_{t_3}(.)}{\star}\ldots\stackrel{C_{t_{p-1}}(.)}{\star}C_{t_{p-1},t_p}. \IEEEeqnarraynumspace
\end{IEEEeqnarray}
Provided that the integral exists for all $\mathbf{x}$, $\mathbf{y}$, $\mathbf{z}$, 
the operator $\stackrel{C(.)}{\star}$ is defined by
\begin{equation}
(A\stackrel{C(\mathbf{z})}{\star}B)(\mathbf{x},\mathbf{y}) = \int_{0}^{\mathbf{z}} A_{,C}(\mathbf{x},\mathbf{r})\cdot B_{C,}(\mathbf{r},\mathbf{y})C(d\mathbf{r}),
\end{equation}
where $A$ is a $(k + n)$-dimensional copula, $B$ is a $(n + l)$-dimensional
copula, $C$ is a $n$-dimensional copula, and $A(\bm{x},d\bm{y}) =A_{,C}(\mathbf{x}, \mathbf{y})C(d\bm{y})$ and $B(d\bm{x},\bm{y})=B_{C,}(\mathbf{x}, \mathbf{y})C(d\bm{x})$ are respectively the derivative of the copula $A(\mathbf{x}, .)$ and $B(., \mathbf{y})$ with respect to the copula $C$.
$A_{,C}$ and $B_{C,}$ are well-defined.
Specifically, for $1$-dimensional Markov process, the copula is expressed by \cite{darsow1992copulas}
\begin{equation}
C_{t_1\ldots{t_n}} = C_{t_1 t_2}\star C_{t_2 t_3}\star\ldots\star C_{t_{n-1}t_n},
\end{equation}
where $C_{t_1\ldots{t_n}}$ is the copula of $\left(X_{t_1},\ldots,X_{t_n}\right)$, $C_{t_{k-1}t_k}$ is the copula of $\left( X_{t_{k-1}}, X_{t_k} \right)$, and $A\star{B}$ is defined as
\begin{multline}
A\star B \left( x_1,\ldots,x_{m+n-1} \right) = \\
\int\limits_{0}^{x_m} 
\frac{\partial A_{,m}(x_1,\ldots,x_{m-1},\xi)}{\partial{\xi}} \frac{\partial B_{1,}(\xi,x_{m+1},\ldots,x_{m+n-1})}{\partial\xi} d\xi,
\end{multline}
for $m$-dimensional copula $A$ and $n$-dimensional copula $B$.

\begin{proposition}
If the dependence in capacity is driven by a Markov process and the instantaneous capacity has a specific distribution with respect to a specific state transition, the additive capacity together with the underlying Markov process, i.e., cumulative capacity, form a Markov additive process, which is a bivariate process with strong Markov property and the increments are conditionally independent given a realization of the underlying Markov process.
\end{proposition}

In other words, a Markov additive process is a non-stationary additive process defined on a Markov process \cite{ccinlar1972markovi,ccinlar1972markovii}, if the Markov process has only one state, it reduces to a Markov process with additive increments.
Since there is no requirement on the $1$-dimensional marginal distribution $X_t^i$ for $\bm{X}$ to be Markov, starting with a Markov process, a multitude of other Markov processes can be constructed by just modifying the marginal distributions \cite{darsow1992copulas,overbeck2015multivariate}. 

\subsection{Distribution Bound}

A Markov additive process is defined as a bivariate Markov process $\{X_t\}=\{(J_t,S(t))\}$ where $\{J_t\}$ is a Markov process with state space $E$ and the increments of $\{S(t)\}$ are governed by $\{J_t\}$ in the sense that \cite{asmussen2010ruin}
\begin{equation}
{E}\left[f(S({t+s})-S(t))g(J_{t+s})|\mathscr{F}_t\right] = {E}_{J_t,0}\left[f(S(s))g(J_s)\right].
\end{equation}
We focus on the finite state space scenario and the structure is fully understood in discrete-time and continuous time settings.

In discrete time, a Markov additive process is specified by the measure-valued matrix (kernel) $\mathbf{F}(dx)$ whose $ij$th element is the defective probability distribution 
\begin{equation}
F_{ij}(dx)={P}_{i,0}(J_1=j,Y_1\in{dx}),
\end{equation}
where $Y_t=S(t)-S(t-1)$. An alternative description is in terms of the transition matrix $\mathbf{P}=(p_{ij})_{i,j\in{E}}$, $p_{ij}={P}_i(J_1=j)$, and the probability measures
\begin{equation}
H_{ij}(dx) = {P}(Y_1\in{dx}|J_0=i, J_1=j) = \frac{F_{ij}(dx)}{p_{ij}}.
\end{equation} 
With respect to a transition probability $p_{ij}$, the increment of $\{S_t\}$ has a distribution $B_{ij}$.

In continuous time, $\{J_t\}$ is a Markov jump process specified by the intensity matrix $\bm{\Lambda}=(\lambda_{ij})_{i,j\in{E}}$. 
When $\{J_t\}$ jumps from $i$ to $j\neq{i}$, the jump of $\{S_t\}$ has a probability $q_{ij}$ with a distribution $B_{ij}$.
When $J_t\equiv i$, $\{ S_t \}$ evolves like a L\'{e}vy process with a characteristic triplet $(\nu_i,\mu_i,\sigma_i^2)$, where $\nu_i\in{R}$, $\sigma_i\ge{0}$, and $\nu_i$ is a nonnegative measure on $R$, if the L\'{e}vy measure $\nu_i$ satisfies, $\int_{-\epsilon}^{\epsilon} |x| \nu_i (dx) < \infty$ and $\nu_i([-\epsilon,\epsilon]^c) = \int_{\{ x:|x|>\epsilon \}} \nu_i(dx)< \infty$, $\forall\epsilon>0$, the L\'{e}vy exponent is expressed as
\begin{equation}
\kappa^{(i)}(\theta) = \theta\mu_i + \theta^2{\sigma_i^2}/2 + \int_{\infty}^{\infty}\left[e^{\theta{x}}-1\right]\nu_i(dx),
\end{equation}
where $\theta\in\Theta=\{ \theta\in{C}: E{e^{\mathcal{R}(\theta)S_1}}<\infty \}$.

Consider the matrix $\widehat{\textbf{F}}_t[\theta]=(E_i[e^{\theta{S(t)}};J_t=j])_{i,j\in{E}}$. 
In discrete time,
\begin{equation}
\widehat{\textbf{F}}_t[\theta]=\widehat{\textbf{F}}[\theta]^t,
\end{equation}
where $\widehat{\textbf{F}}[\theta]=\widehat{\textbf{F}}_1[\theta]$ is a $E\times{E}$ matrix with $ij$th element $\widehat{F}^{(ij)}[\theta]=p_{ij}\int{e^{\theta{x}}F^{(ij)}}(dx)$, and $\theta\in\Theta=\{ \theta\in{R}:\int{e^{\theta{x}}F^{(ij)}}(dx)<\infty \}$  \cite{asmussen2003applied}. 
In continuous time, 
\begin{equation}
\widehat{\textbf{F}}_t[\theta] = e^{t\bm{K}[\theta]},
\end{equation}
where $\bm{K}[\theta] = \Lambda + \left( \kappa^{(i)}(\theta) \right)_{\text{diag}} + \lambda_{ij}q_{ij}\left(\widehat{B}_{ij}[\theta]  -1\right)$ \cite{asmussen2003applied}.
By Perron-Frobenius theorem, $\widehat{\textbf{F}}[\theta]$ has a positive real eigenvalue with maximal absolute value, $e^{\kappa(\theta)}$, in discrete time; $\bm{K}[\theta]$ has a real eigenvalue with the maximal real part, $\kappa(\theta)$, in continuous time.
The corresponding right and left eigenvectors are respectively $\textbf{h}^{(\theta)}=\left(h_{i}^{(\theta)}\right)_{i\in{E}}$ and $\textbf{v}^{(\theta)}=\left(v_{i}^{(\theta)}\right)_{i\in{E}}$, particularly, $\textbf{v}^{(\theta)}$, $\textbf{v}^{(\theta)}\textbf{h}^{(\theta)}=1$ and $\bm{\pi}\textbf{h}^{(\theta)}=1$, where $\bm{\pi}=\textbf{v}^{(0)}$ is the stationary distribution and $\textbf{h}^{(0)}=\textbf{e}$.

With an exponential change of measure, a likelihood ratio process is formulated \cite{asmussen2003applied},
\begin{equation}
L(t) = \frac{h^{(\theta)}(J_t)}{h^{(\theta)}(J_0)}e^{\theta{S(t)}-t\kappa(\theta)},
\end{equation}
which is a mean-one martingale. This martingale process is useful for performance analysis of the wireless channel capacity. The distribution function results are as follows.

\begin{theorem}\label{theorem_distribution_bound}
For a Markov additive process, conditional on initial state $J_0$, the distribution of the cumulative capacity is expressed as, for some $\theta>0$,
\begin{equation}
1 - \frac{h^{(\theta)}(J_0) e^{t\kappa(\theta)-\theta{x}} }{\min\limits_{j\in{E}}(h^{(\theta)}(J_j))} \le F_{S(t)}(x) \le \frac{h^{(-\theta)}(J_0) e^{t\kappa(-\theta)+\theta{x}} }{\min\limits_{j\in{E}}(h^{(-\theta)}(J_j))},
\end{equation}
and the distribution of the transient capacity is expressed as
\begin{equation}
1 - \frac{h^{(\theta)}(J_0){e^{-y_l}}}{\min\limits_{j\in{E}}(h^{(\theta)}(J_j))} \le P\left\{ \overline{C}(t) \le c^\ast \right\}\le \frac{h^{(-\theta)}(J_0){e^{-y_u}}}{\min\limits_{j\in{E}}(h^{(-\theta)}(J_j))},
\end{equation}
where $c^\ast=\frac{t\kappa(\theta^\ast)+y^\ast}{\theta^\ast{t}}$, with $y^\ast=y_u$ for $\theta^\ast<{0}$ for the upper bound, and $y^\ast=y_l$ for $\theta^\ast>{0}$ for the lower bound.
\end{theorem}

The theorem indicates that the distribution of wireless channel capacity with Markov dependence is light-tailed.
Analytical and computational results of transient capacity are shown in Fig. \ref{transient_capacity_markov_additive}.

\begin{figure}[!t]
\centering
\includegraphics[width=3.5in]{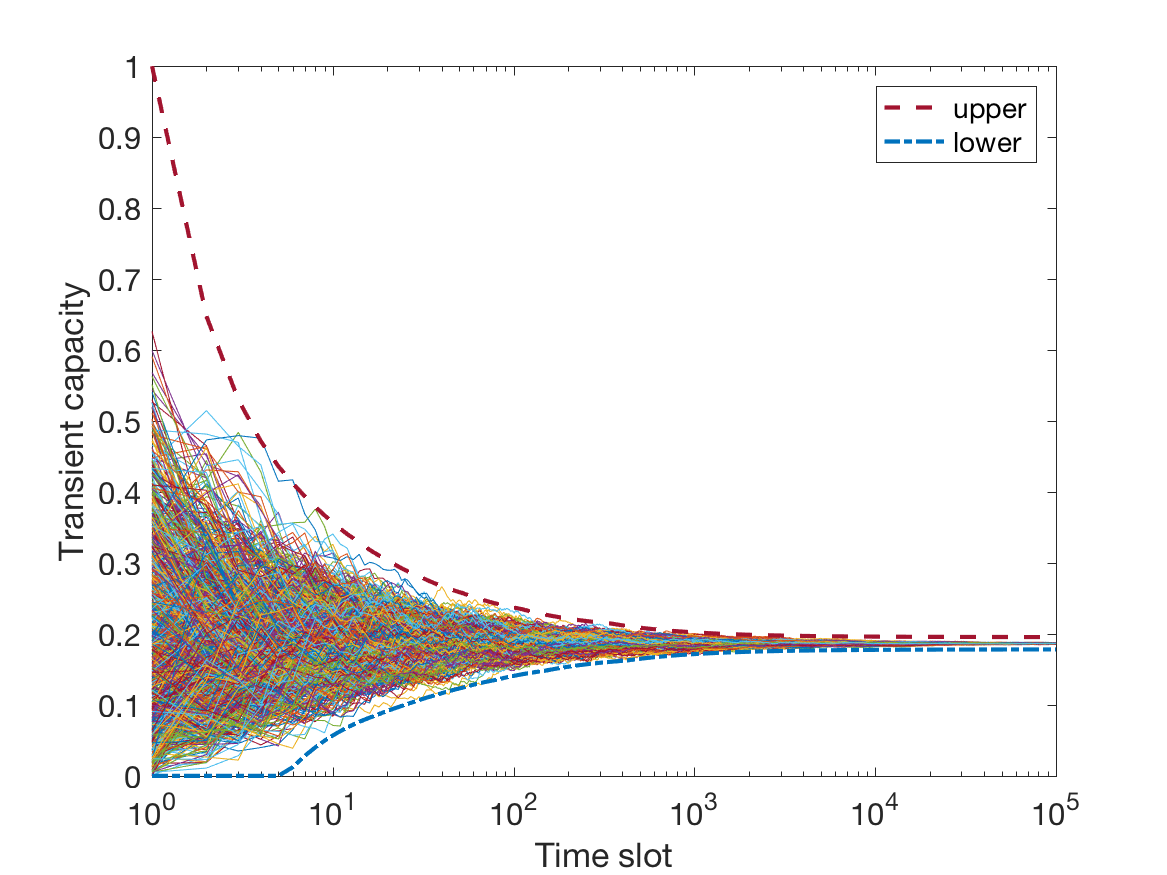}
\caption{Transient capacity of Markov additive Rayleigh channel. According to the strong law of large numbers extended to the Markov additive process, the transient capacity converges to the mean as time goes to infinity, i.e., the convergence of sample paths. The large deviation results are upper bound and lower bound with respect to the mean. Results are normalized, with violation probability $\epsilon=10^{-3}$, $W=20$kHz, $\textbf{SNR}=[e^{0.5}~e^{0.5}; 0.7e^{0.5}~0.7e^{0.5}]$ and $\textbf{P}=[0.3~0.7; 0.1~0.9]$, and 1000 sample paths.}
\label{transient_capacity_markov_additive}
\end{figure}

\section{Dependence Control}\label{control}

\subsection{Performance Measure}

The wireless channel is essentially a queueing system with cumulative service process $S(t)$ and cumulative arrival process $ A(0,t)=\sum\limits_{s=0}^{t}a(s)$, 
where $a(t)$ denotes the traffic input to the channel at time $t$,
and the temporal increment in the system is expressed as
\begin{equation}
X(t) = a(t)-C(t).
\end{equation} 
The queueing principle of the wireless channel is expressed through the backlog in the system, which is a reflected process of the temporal increment $X(t)$ \cite{asmussen2003applied}, i.e.,
\begin{equation}
B(t+1) = \left[ B(t) + X(t)\right]^{+},
\end{equation}
assume $B(0)=0$,
the backlog function is then expressed as
\begin{equation}\label{eq-bl}
B(t) = \sup_{0\le{s}\le{t}}({A}(s,t)-{S}(s,t)). 
\end{equation}
For a lossless system, the output is the difference between the input and backlog, $A^{\ast}(t) = A(t) - B(t)$, i.e., 
\begin{equation}\label{eq-ior}
A^{\ast}(t) = A\otimes S(t),
\end{equation}
where $f\otimes g(t) = \inf_{0\le{s}\le{t}}\{f(s)+g(s,t)\}$ is the bivariate min-plus convolution \cite{baccelli1992synchronization},
and the delay is defined via the input-output relationship \cite{ciucu2014sharp}, i.e., 
\begin{equation}
D(t) = \inf\left\{ d\ge{0}: A(t-d)\le A^\ast(t) \right\},
\end{equation}
which is the virtual delay that a hypothetical arrival has experienced on departure.
The maximum rate of traffic with delay requirement that the system can support without dropping is defined as the delay-constrained capacity or throughput \cite{yuehong2012analysis}:
\begin{equation}
\overline{C}{(d,\epsilon)} = \sup_{P(D(t) > d) \le \epsilon, \forall t} E\left[ \frac{A(t)}{t} \right].
\end{equation}
The above results also apply to continuous-time setting.

To embody the impact of dependence in wireless channels, 
we assume that the input is a constant fluid process, i.e., 
\begin{equation}
A(t) = \lambda{t}.
\end{equation}
The results of performance metrics are summarized in the following theorem and corollary.

\begin{theorem}\label{theorem_performance_bound}
Consider a constant arrival process $A(t)=\lambda{t}$, the delay conditional on the initial state $J_0=i$ is bounded by
\begin{equation}
\frac{h^{(-\theta)}(J_i) e^{-\theta{\lambda{d}}} }{\max\limits_{j\in{E}}h^{(-\theta)}(J_j)} \le P_i(D\ge{d}) \le \frac{h^{(-\theta)}(J_i) e^{-\theta{\lambda{d}}} }{\min\limits_{j\in{E}}h^{(-\theta)}(J_j)},
\end{equation}
where $-\theta$ is the negative root of $\kappa(\theta)=0$ of $S(t)-\lambda{t}$ and $\bm{h}^{(-\theta)}$ is the corresponding right eigenvector,
given the initial state distribution $\bm{\varpi}$, the delay and backlog are bounded by 
\begin{eqnarray}
P(D\ge d) &=& \sum_{i\in{E}}\varpi_{i}P_i(D\ge d), \\
P(B\ge b) &=& P(D\ge b/\lambda). 
\end{eqnarray}
\end{theorem}

\begin{corollary}
For constant fluid traffic $A(t)=\lambda{t}$, the delay-constrained capacity, letting $P(D\ge d)=\epsilon$, is expressed as
\begin{equation}
\frac{-1}{\theta{d}}{\log\frac{\epsilon\cdot{\max\limits_{j\in{E}}h^{(-\theta)}(J_j)}}{\sum_{i\in{E}}{\varpi_i}h^{(-\theta)}(J_i)}} \le \lambda \le \frac{-1}{\theta{d}}{\log\frac{\epsilon\cdot{\min\limits_{j\in{E}}h^{(-\theta)}(J_j)}}{\sum_{i\in{E}}{\varpi_i}h^{(-\theta)}(J_i)}}.
\end{equation}
\end{corollary}

\begin{proof}
Consider the delay-constrained capacity for the constant fluid process $A(t) =\lambda{t}$,
\begin{equation}
\overline{C}{(d,\epsilon)} = \sup_{P(D(t) > d) \le \epsilon, \forall t} \lambda,
\end{equation}
the result follows directly from Theorem \ref{theorem_performance_bound}.
\end{proof}

\begin{remark}
It's a folk law that the regularity of arrival or service processes results in better performance measures, and it's been proved that for some involved system the queue length of a constant fluid input is the shortest for all types of inputs that have the same average traffic rate \cite{muller2002comparison}, thus the minimal delay and maximal delay-constrained capacity.
\end{remark}

\subsection{Dependence Classification}

We classify the dependence into three types, i.e., positive dependence, independence, and negative dependence.
Intuitively, positive dependence implies that large or small values of random variables tend to occur together, while negative dependence implies that large values of one variable tend to occur together with small values of others \cite{denuit2006actuarial}.

We use discrete-time setting in this subsection.
Formally, the cumulative capacity $S$ is said to have a positive dependence structure $S_{+}$ in the sense of increasing convex order, if
\begin{equation}
S_{\perp} \le_{icx} S_{+},
\end{equation}
or a negative dependence structure $S_{-}$ in the sense of increasing convex order, if
\begin{equation}
S_{-} \le_{icx}S_{\perp},
\end{equation}
where $S_{\perp}$ has an independence structure. 
Since the mean of sum of random variables equals the sum of means of individual random variables, i.e.,
\begin{equation}
E[S_{-}]= E[S_{\perp}]= E[S_{+}],
\end{equation}
the increasing convex ordering and convex ordering of cumulative capacity are equivalent \cite{kampke2015income}, i.e.,
\begin{multline}
S_{-} \le_{icx} S_{\perp} \le_{icx}S_{+} \iff
S_{-} \le_{cx} S_{\perp} \le_{cx}S_{+}.
\end{multline}
It's worth noting that the supermodular ordering of instantaneous capacity, i.e.,
\begin{equation}
\textbf{C} \le_{sm} \widetilde{\textbf{C}},
\end{equation}
indicates that the marginal distributions of the instantaneous increments are identical, 
particularly, if $\textbf{C}\le_{sm}\widetilde{\textbf{C}}$, then $\sum_{i=1}^{n}C(i)\le_{cx}\sum_{i=1}^{n}\widetilde{C}(i)$.

As the distribution of wireless channel capacity is light-tailed, the asymptotic behavior of the bounding function is still exponential for weak forms of dependence and becomes heavy-tailed for stronger dependence \cite{asmussen2010ruin}. The ordering of the exponential adjustment coefficient is as follows.

\begin{theorem}
Consider two wireless channel capacity processes, if the cumulative capacities are convex ordered, then the adjustment coefficients for the delay bounds are correspondingly ordered, i.e.,
\begin{equation}
S(t)\le_{cx}\widetilde{S}(t),\ \forall{t}\in\mathbb{N}
\implies
\widetilde{\theta} \le \theta.
\end{equation}
\end{theorem}

\begin{proof}
Consider the negative increment process, i.e.,
\begin{equation}
-X(t)=C(t)-a(t). 
\end{equation}
If it is light-tailed, then the delay violation probability has an exponential bound with adjustment coefficient $\theta>0$ defined by $\kappa(\theta)=0$, where \cite{asmussen2010ruin,muller2001asymptotic}
\begin{equation}
\kappa(\theta) = \lim_{t\rightarrow\infty}\frac{1}{t}E\left[ e^{\theta\sum_{i=1}^{t}{X(i)}} \right].
\end{equation}
By exploring the ordering of the cumulative increment process, 
\begin{equation}
\sum_{i=1}^{n}-X(i) \le_{cx} \sum_{i=1}^{n}-\widetilde{X}(i),\ \forall{n}\in\mathbb{N},
\end{equation}
the adjustment coefficients are ordered as follows \cite{asmussen2010ruin,muller2001asymptotic}
\begin{equation}
\widetilde{\theta} \le \theta.
\end{equation}
Specifically, for constant arrival, the ordering of the cumulative capacity results in the ordering of the cumulative negative increment process. 
\end{proof}

The ordering of the adjustment coefficients gives an ordering of the asymptotic delay tail distribution, with some restrictions, the result can be applied to the delay-constrained capacity.

\begin{corollary}
For delay bounding functions with the same prefactor or are bounded by a same prefactor before the exponential term, the ordering of the cumulative capacity $S_{-} \le_{cx} S_{\perp} \le_{cx}S_{+}$ indicates the ordering of the delay, i.e.,
\begin{equation}
P\left(D_{-}\ge{x}\right) \le P\left(D_{\perp}\ge{x}\right) \le P\left(D_{+}\ge{x}\right),
\end{equation}
and the ordering of the delay-constrained capacity for the same prefactor, i.e.,
\begin{equation}
\lambda_{-} \ge \lambda_{\perp} \ge \lambda_{+}.
\end{equation}
\end{corollary}

The impact of negative dependence and positive dependence in capacity on delay and comparison with independence in capacity are illustrated in Fig. \ref{delay_markov_negative_positive_dependence}.
We fix the noise power density and change the transmission power in SNR. The result shows that the wireless channel attains a better performance with less transmission power or smaller capacity for negative dependence in contrast to positive dependence.

\begin{figure}[!t]
\centering
\includegraphics[width=3.5in]{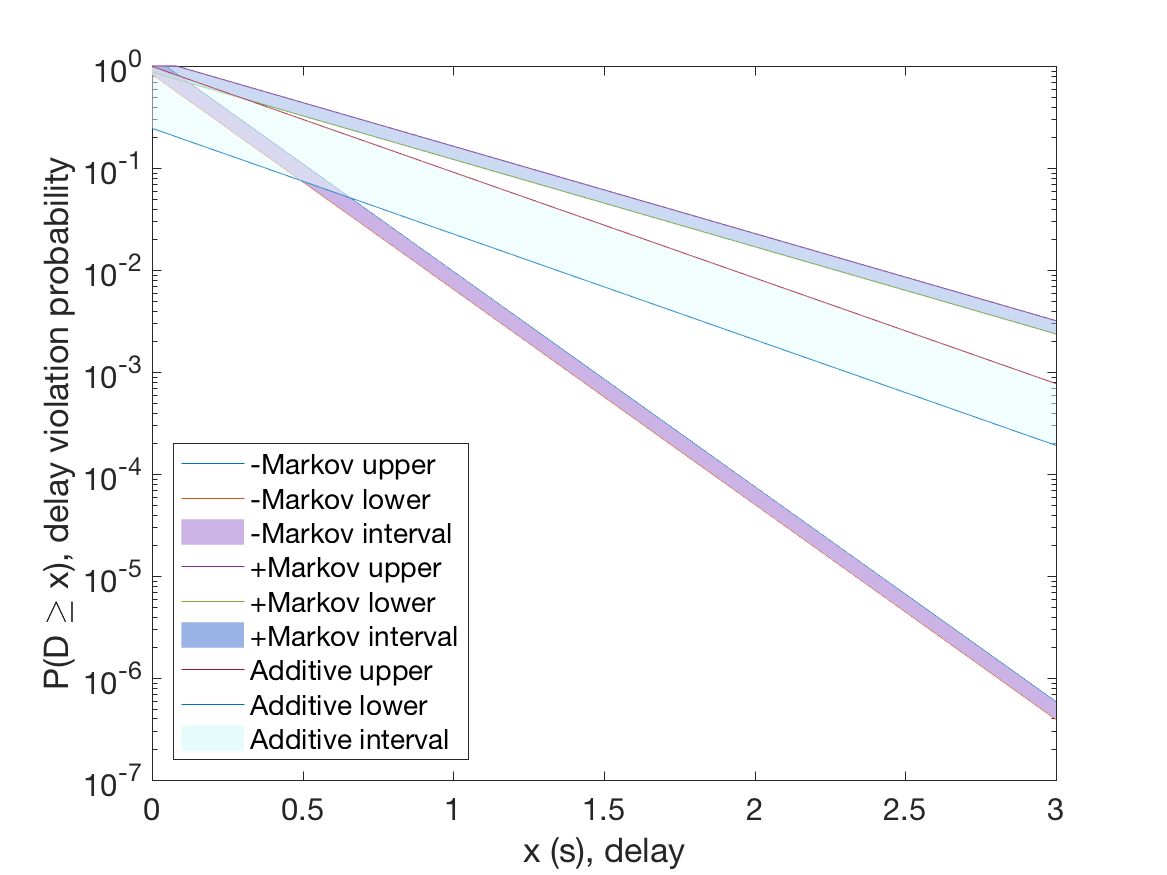}
\caption{Delay tail distribution of Rayleigh channel. ``-'' and ``+'' depict respectively negative and positive dependence in capacity, the lines depict the double-sided bounds with the intervals depicted as the shaded areas. $\lambda=10$kbits, $W=20$kHz, $SNR=e^{0.5}$ for the independence case of additive capacity process, $\textbf{SNR}=[e^{0.5}~e^{0.5}; 0.7e^{0.5}~0.7e^{0.5}]$, and 
$\textbf{P}=[0.4125~0.5875; 0.2518~0.7482]$ calculated from Fr\'{echet} copula with $\alpha=0.5$ for $\lambda-C(t)$ indicating negative dependence in capacity and 
$\textbf{P}=[0.2875~0.7125; 0.3054~0.6946]$ calculated from Fr\'{echet} copula with $\alpha=-0.5$ for $\lambda-C(t)$ indicating positive dependence in capacity, 
for the dependence case of Markov additive capacity process with initial distribution $\bm{\varpi}=[0.5~0.5]$ and stationary distribution $\bm{\pi}=[0.3\ 0.7]$. 
}
\label{delay_markov_negative_positive_dependence}
\end{figure}

\subsection{Dependence under Control}

We distinguish the random parameters in the wireless system, which cause the dependence in the wireless channel capacity, into two categories, i.e., uncontrollable parameters and controllable parameters.
Uncontrollable parameters represent the property of the environment that can not be interfered, e.g., fading, 
while controllable parameters represent the configurable property of the wireless system, e.g., power. 
We use the controllable parameters to induce negative dependence into the wireless channel capacity to achieve dependence control.

We assume no Granger causality among random parameters.
No-Granger causality is a concept initially introduced in econometrics and refers to a multivariate dynamic system in which each variable is determined by its own lagged values and no further information is provided by the lagged values of other variables \cite{cherubini2011copula}. 

\begin{proposition}
For a $n$-dimensional process $\bm{X}$, $\bm{X}^1,\ldots,\bm{X}^{i-1}$, $\bm{X}^{i+1},\ldots,\bm{X}^n$ do not Granger cause $\bm{X}^i$, if \cite{cherubini2010copula,cherubini2011copula}
\begin{equation}
P\left( X^{i}_{t_{k+1}} \le x | \mathscr{F}_{t_k}^{\bm{X}^1,\ldots,\bm{X}^n} \right) = P\left( X^{i}_{t_{k+1}} \le x | \mathscr{F}_{t_k}^{\bm{X}^i} \right).
\end{equation}
\end{proposition}

No-Granger causality and Markov property of each process with respect to its natural filtration together imply the Markov structure of the system as a whole \cite{cherubini2010copula,cherubini2011copula}.
Additional restriction is required for the converse to hold, the $2$-dimensional result is available in \cite{cherubini2011copula}, and the following theorem is an extension to $n$-dimensional case.

\begin{theorem}\label{theorem_no_granger_causality_biset}
For a $n$-dimensional Markov process $\bm{X}$ consisting two dimension sets $\overline{\bm{X}}$ and $\underline{\bm{X}}$, $\bm{X}= \overline{\bm{X}} \cup \underline{\bm{X}}$, $\overline{\bm{X}}$ does not Granger cause $\underline{\bm{X}}$, if and only if
\begin{multline}
C_{j,j+1} \left(\bm{u}_{\underline{\bm{X}}_j}, \bm{u}_{\overline{\bm{X}}_j}, \bm{u}_{\underline{\bm{X}}_{j+1}}, \bm{1}_{\bm{u}_{\overline{\bm{X}}_{j+1}}} \right)  \\
= C_{\overline{\bm{X}}_j\underline{\bm{X}}_j} \stackrel{C_{\underline{\bm{X}}_j} \qty(\bm{u}_{\underline{\bm{X}}_j})}{\star} C_{\underline{\bm{X}}_j\underline{\bm{X}}_{j+1}} \left(\bm{u}_{\overline{\bm{X}}_j}, \bm{u}_{\underline{\bm{X}}_{j+1}} \right),
\end{multline}
$\underline{\bm{X}}$ does not Granger cause $\overline{\bm{X}}$, if and only if
\begin{multline}
C_{j,j+1} \left(\bm{u}_{\underline{\bm{X}}_j}, \bm{u}_{\overline{\bm{X}}_j}, \bm{1}_{\bm{u}_{\underline{\bm{X}}_{j+1}}}, {\bm{u}_{\overline{\bm{X}}_{j+1}}} \right)  \\
= C_{\underline{\bm{X}}_j\overline{\bm{X}}_j} \stackrel{C_{\overline{\bm{X}}_j} \qty(\bm{u}_{\overline{\bm{X}}_j})}{\star} C_{\overline{\bm{X}}_j\overline{\bm{X}}_{j+1}} \left(\bm{u}_{\underline{\bm{X}}_j}, \bm{u}_{\overline{\bm{X}}_{j+1}} \right).
\end{multline}
\end{theorem}

\begin{remark}
Specifically, for the wireless channel capacity that is modeled by a multivariate Markov process, let $\overline{\bm{X}}$ and $\underline{\bm{X}}$ represent respectively the uncontrollable and controllable parameters.
The no-Granger causality guarantees that if the uncontrollable and controllable parameters form a multivariate Markov process, the processes of the uncontrollable and controllable parameters are also Markov processes, which is necessary in dependence control because we need to model the uncontrollable parameters with a certain process and to configure the controllable parameters in a certain way based on a certain process.
\end{remark}

A stronger restriction is that all the $1$-dimensional Markov processes do not Granger cause each other, and the results are as follows.

\begin{theorem}\label{theorem_no_granger_causality}
For a $n$-dimensional Markov process $\bm{X}$ with temporal copula $C_{j,j+1}$ and spatial copula $C_j$,
\begin{IEEEeqnarray}{rCl}
P\left( X^{i}_{t_{k+1}} \le x | {\bm{X}^1_{t_k},\ldots,\bm{X}^n_{t_k}} \right) = P\left( X^{i}_{t_{k+1}} \le x | {\bm{X}^i_{t_k}} \right), \IEEEeqnarraynumspace
\end{IEEEeqnarray}
if and only if
\begin{multline}
C_{j,j+1} \left( x_j^{1},\ldots,x_j^n,1,\ldots,x_{j+1}^i,\ldots,1 \right) = \\
C_j^{,i} \star C_{j,j+1}^i \left( x_j^1,\ldots,x_j^{i-1},x_j^{i+1},\ldots,x_j^n,x_j^i,x_{j+1}^i \right),
\end{multline}
where $C_j^{,i}$ is the reordered spatial copula, and $C_{j,j+1}^i$ is the temporal copula of the $1$-dimensional Markov process $\bm{X}^i$.
\end{theorem}

\begin{proof}
The proof follows analogically from Theorem \ref{theorem_no_granger_causality_biset}.
\end{proof}

\begin{example}
For a $2$-dimensional Markov process $\bm{X}$, $\bm{X}^2$ does not Granger cause $\bm{X}^1$, if and only if \cite{cherubini2011copula}
\begin{equation}
C_{j,j+1}(u_1,v_1,u_2,1) = C_{X^2_j,X^1_j} \star C_{X^1_j,X^1_{j+1}}(v_1,u_1,u_2),
\end{equation}
and $\bm{X}^1$ does not Granger cause $\bm{X}^2$, if and only if \cite{cherubini2011copula}
\begin{equation}
C_{j,j+1}(u_1,v_1,1,v_2) = C_{X^1_j,X^2_j} \star C_{X^2_j,X^2_{j+1}}(u_1,v_1,u_2).
\end{equation}
In the special case, if the spatial dependence is expressed by the product copula, then 
\begin{eqnarray}
C_{j,j+1}(u_1,v_1,u_2,1) &=& v_1 C_{X^1_{j}X^1_{j+1}} \left( u_1,u_2 \right),\\
C_{j,j+1}(u_1,v_1,1,v_2) &=& u_1 C_{X^2_{j}X^2_{j+1}} \left( v_1,v_2 \right).
\end{eqnarray}
\end{example}

Since the copula requires continuity by definition, interpolation is needed to construct a copula from the transition matrix of a Markov process \cite{darsow1992copulas}, while it's not needed to calculate the transition matrix from a copula.
The approach to calculate the transition probability of a Markov chain given the copula of the two consecutive levels is summarized in the following theorem.

\begin{theorem}
For a $1$-dimensional Markov process with finite state space $E$ and initial distribution $\bm\varpi$, given the copula between successive levels $C_{j,j+1}$, 
\begin{equation}
\sum_{s_j\le{\bm x}}\bm{\varpi}_{j}({s_j})\bm{P}_j(s_j,s_{j+1}\le\bm{y}) = C_{j,j+1}\left( F_j(\bm{x}), F_{j+1}(\bm{y}) \right),
\end{equation}
where $\bm{x}$ and $\bm{y}$ are the ordered state space vector, the state distribution at $j$ is $\bm{\varpi}_j=\bm{\varpi}\prod_{0\le{k}\le{j}}\bm{P}_k$, and $F_j(s_j) = \sum\bm{\varpi}_j(s_k\le{s_j})$ and $F_{j+1}=\bm{\varpi}_j{\bm{P}_j}$.
Together with the unity property of transition matrix $\sum_{j\in{E}}p_{ij}=1$, $\forall{i}\in{E}$, the transition probabilities $\bm{P}_j$ are obtained.
\end{theorem}

\begin{proof}
For random variables $X$ and $Y$ with the copula $C$ \cite{darsow1992copulas}
\begin{IEEEeqnarray}{rCl}
E\left( I_{Y<y}|X \right)(\omega) &=& C_{1,}\left( F_X(X(\omega)), F_Y(y) \right)\ a.s.,
\end{IEEEeqnarray}
it follows that 
\begin{equation}
C\left( F_s(\bm{x}), F_t(\bm{y}) \right) = \int_{-\infty}^x P\left( X_t\le{y} | X_s = \xi \right) d\xi.
\end{equation}
The result directly follows.
\end{proof}

\begin{example}
For a $2$-state homogeneous Markov process, the equations are expressed as
\begin{IEEEeqnarray}{rrCl}
  \IEEEyesnumber\IEEEyessubnumber*
  & C\left(F(0),F(0)\right) & = & \pi_0 p_{00},
  \\*
  & C\left(F(1),F(0)\right) & = & \pi_0 p_{00} + \pi_1 p_{10},
  \\*[-0.625\normalbaselineskip]
  \smash{\left\{
      \IEEEstrut[11\jot]
    \right. } \nonumber
\\*[-0.625\normalbaselineskip]
  & C\left(F(1),F(1)\right) & = & \pi_0 \left(p_{00} + p_{01} \right) + \pi_1 \left(p_{10} + p_{11} \right), \IEEEeqnarraynumspace
  \\*
  & C\left(F(0),F(1)\right) & = & \pi_0 \left(p_{00} + p_{01} \right).
\end{IEEEeqnarray}
Given a stationary distribution $[\pi_{0}\ \pi_{1}]$, $F(0)=\pi_0$ and $F(1)=\pi_0+\pi_1$, we obtain the values of $p_{00}$ and $p_{10}$ from the equations, and we further obtain $p_{01}=1-p_{00}$ and $p_{11}=1-p_{10}$ from the unity property.
\end{example}

\begin{algorithm}[!t]
 \caption{Algorithm for Dependence Control}
 \begin{algorithmic}[1]
 \label{dependence_control_algorithm}
 \renewcommand{\algorithmicrequire}{\textbf{Model:}}
 \renewcommand{\algorithmicensure}{\textbf{Result:}}
 \REQUIRE A $n$-dimensional Markov process consisting of $n$ $1$-dimensional Markov processes without Granger causality
 \ENSURE  Transition matrix of the controllable parameter
 \\ 
  \STATE Initialisation: $C_{j,j+1}$, $C_j$, and $\bm{\varpi}$
  \FOR {$j = 0$ to $t-1$}
  \FOR {$1\le i\le{n}$ of interest} 
  \STATE Calculate $\bm{P}_j^i$ and $\bm{\varpi}_{j+1}^i=\bm{\varpi}^i_j{\bm{P}^i_j}$, with\\ $\sum \bm{\varpi}_j^i{\bm{P}_j^i} = C_{j,j+1}^i$
  \ENDFOR
  \ENDFOR
 \RETURN $\bm{P}$ 
 \end{algorithmic} 
\end{algorithm}

The algorithm of dependence control is shown in Algorithm \ref{dependence_control_algorithm}.
It's worth noting that the Markov property is a pure property of copula,
different copula functions provide a way to character the negative or positive dependence, based on which we can calculate the transition matrix of the controllable parameters in the wireless system, e.g., power, and bring their impacts into capacity.
The Markov family copula is elaborated in Appendix \ref{markov_family_copula}. 

An example is illustrated in Fig. \ref{scatter_hist_capacity}. 
The fading process is independent and the power changes with negative or positive dependence, the result shows that the times series of the instantaneous capacity exhibits weakly negative dependence or weakly positive dependence, and the impact is manifested in the transient capacity, on the other hand, it indicates that the impact of independent parameters is strong.

\begin{figure}[!t]
\centering
\includegraphics[width=3.5in]{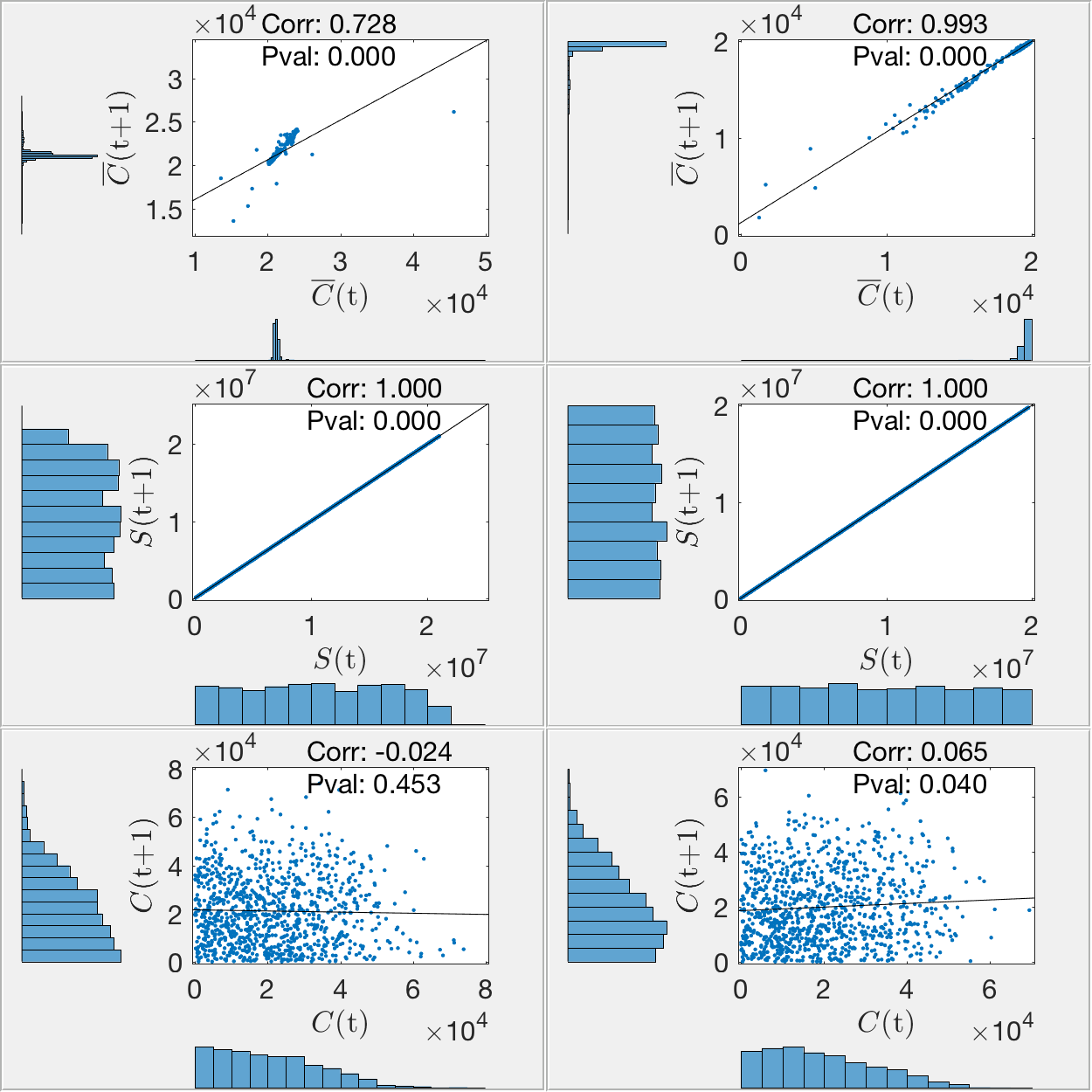}
\caption{Wireless channel capacity of Markov additive Rayleigh channel. 
The uncontrollable parameter is fading with one state and the controllable parameter is power with two states. 
The Markov process is time homogeneous without Granger causality.
The dependence structure is Gaussian copula with correlation matrix $\bm{\Sigma}=[1\ 0.1\ -0.5\ 0; 0.1\ 1\ 0\ 0; -0.5\ 0\ 1\ 0.1; 0\ 0\ 0.1\ 1]$ as negative dependence (left column), $\bm{\Sigma}=[1\ 0.1\ 0.5\ 0; 0.1\ 1\ 0\ 0; 0.5\ 0\ 1\ 0.1; 0\ 0\ 0.1\ 1]$ as positive dependence (right column), initial distribution $\bm{\varpi}= [0.5\ 0.5]$, stationary distribution $\bm{\pi}= [0.3\ 0.7]$. $W=20$kHz and $\textbf{SNR}=[e^{0.5}~e^{0.5}; 0.7e^{0.5}~0.7e^{0.5}]$. $1000$ time slots.
Correlation coefficient and probability value between the time series and lag-1 series are provided.
}
\label{scatter_hist_capacity}
\end{figure}

\section{Conclusion}\label{conclusion}

We discovered a hidden resource in wireless channels, namely stochastic dependence. Specifically, when the wireless channel capacity evolves with negative dependence, a better channel performance is attained with a smaller capacity.
The contributions of this paper are as follows.
\begin{enumerate}
\item
We modeled the wireless channel capacity as a Markov additive process, used copula to represent the dependence in the underlying Markov process, and derived double-sided distribution function bounds of the cumulative capacity and transient capacity.
Both continuous-time and discrete-time setting were considered.
We treated the wireless channel as a queueing system and derived tail bounds of delay and backlog for a constant fluid arrival process, which is usually regarded as the best performance measure that the wireless channel attains. 
In addition, we obtained a double-sided bound of delay-constrained capacity as the inverse function of delay.
\item
We classified the dependence in wireless channel capacity into three types by  defining stochastic orders on the capacity, namely positive dependence, independence, and negative dependence.
In terms of the performance measure, we proved that the wireless channel attains a better performance when the capacity bears negative dependence relative to positive dependence.
Numerical results showed that a better performance is attained even with a smaller average capacity.
\item
We distinguished the random parameters in the wireless system into two categories, i.e., uncontrollable parameters and controllable parameters.
We assumed that these two types of parameters do not Granger cause each other, for which we provided a sufficient and necessary condition based on the copula of the Markov process.
We proposed to use the controllable parameters to induce negative dependence into the capacity, specifically, we constructed a sequence of temporal copulas of the Markov process based on a priori information of the random parameter processes, from which we calculated the transition matrix of the controllable parameters, given the expected dependence information.
Thus, we achieved dependence control.
\end{enumerate}
It's worth noting that though this paper focuses on dependence with Markov property, the influence of positive and negative dependence holds in general.


\appendices

\section{Copula}\label{copula_introduction}

Consider a joint distribution $F\left( X_1,\ldots,X_n \right)$ with marginal distribution $F_i(x)$, $i=1,\ldots,n$. Denote $u_i = F_i\left(X_i\right)$, which is uniformly distributed in the unit interval, then \cite{embrechts2009copulas}
\begin{IEEEeqnarray}{rCl}
F\left( X_1,\ldots,X_n \right) &=& F\left( F_1^{-1}\left(u_1\right),\ldots, F_n^{-1}\left(u_n\right) \right) \\
&\equiv& C\left(u_1,\ldots,u_n\right),
\end{IEEEeqnarray}
where $C$ is a copula with standard uniform marginals, specifically, if the marginals are continuous, the copula is unique.

\begin{definition}
A $n$-dimensional copula $C$ is a distribution function on $[0, 1]^n$ with standard uniform marginal distributions, if
\begin{enumerate}
\item
$C(u_1,\ldots,u_n)$ is increasing in each component $u_i$;
\item
$C(1,\ldots,1,u_i,1,\ldots,1)=u_i$ for all $i\in\{1,\ldots,n\}$, $u_i\in[0,1]$;
\item
For all $(a_1,\ldots,a_n), (b_1,\ldots,b_n) \in[0,1]^n$ with $a_i\le{b_i}$, 
\begin{equation}
\sum_{i_1=1}^{2}\ldots\sum_{i_n=1}^{2}(-1)^{i_1+\ldots+i_n}C(u_{1,i_1},\ldots,u_{n,i_n})\ge{0}, \nonumber
\end{equation}
where $u_{j,1}=a_j$ and $u_{j,2}=b_j$ for all $j\in\{1,\ldots,n\}$.
\end{enumerate}
\end{definition}

\begin{example}
The extremely positive dependence, independence, and extremely negative dependence are expressed by copulas.
For $2$-dimensional copula, the extremely positive copula, product copula (independence), and extremely negative copula are defined as
\begin{eqnarray}
M(x,y) &=& \min(x,y), \\
P(x,y) &=& xy, \\
W(x,y) &=& \max(x+y-1,0).
\end{eqnarray}
For a $n$-dimensional copula function $C$, the extremely positive copula functional and product copula functional are defined as
\begin{eqnarray}
M(\bm{x},\bm{y}) &=& C(\min(\bm{x},\bm{y})), \\
P(\bm{x},\bm{y}) &=& C(\bm{x})C(\bm{y}),
\end{eqnarray}
where the minimum is coordinate-wise.
Since the extreme negative copula function is not a copula for higher dimensional case,
let $\bm{X}$ be a $n$-dimensional uniform random vector with copula $C$, the extremely negative copula functional is defined as \cite{overbeck2015multivariate}
\begin{equation}
W(\bm{x},\bm{y}) = P(\bm{X}<\bm{x},T(\bm{X})<\bm{y}) = \int_{\bm{0}}^{\bm{x}} \bm{1}_{T(\bm{z})<\bm{y}} C(d\bm{z}),
\end{equation}
where $T:[0,1]^n\rightarrow[0,1]^n$ is a bijective mapping with the following property
\begin{eqnarray}
P(T(\bm{X})<\bm{x}) &=& C(\bm{x}), \\
T^2(\bm{x}) &=& \bm{x}, \\
C(\{\bm{x} : T(\bm{x})\neq \bm{x} \}) &>& 0.
\end{eqnarray}
In case the copula $C$ is symmetric, i.e., $C(\bm{x})=\overline{C}(1-\bm{x})$, the mapping is expressed as $T(\bm{x})=1-\bm{x}$.
\end{example}

The Sklar's theorem depicts that every distribution can be written as a copula function taking marginals as arguments, and every copula function taking arbitrary marginals as arguments is a joint distribution.
In addition, the functional invariance property implicates that the dependence structure represented by a copula is invariant under non-decreasing and continuous transformations of the marginals.

\subsection{Markov Family Copula}\label{markov_family_copula}

To construct a Markov process, the copula must satisfy the Markov property condition \cite{darsow1992copulas,overbeck2015multivariate}.
\begin{definition}
A family of $2n$-dimensional copula $C_{st}$, $s<t$, is called a Markov family, if \cite{overbeck2015multivariate}
\begin{eqnarray}
C_{su}(\bm{1},\cdot) &=& C_{ut}(\cdot,\bm{1}), \\
C_{st} &=& C_{su} \stackrel{C_{u}}{\star} C_{ut},
\end{eqnarray}
for all $s<u<t$.
\end{definition}
Examples of Markov family copula are Gaussian copula and Fr\'{e}chet copula \cite{overbeck2015multivariate}.

\begin{example}
The $n$-dimensional Gausssian copula is expressed as 
\begin{eqnarray}
C_{\bm{\Sigma}}(\bm{u}) = \Phi_{\bm{\Sigma}}\left( \Phi^{-1}(u_1),\ldots,\Phi^{-1}(u_n) \right),
\end{eqnarray}
where $\Phi_{\bm{\Sigma}}$ denotes the joint distribution of the $n$-dimensional standard normal distribution with linear correlation matrix $\bm{\Sigma}$, and $\Phi^{-1}$ denotes the inverse of the distribution function of the $1$-dimensional standard normal distribution.
\end{example}

\begin{example}
A convex combination of $M$, $P$, and $W$ is a Markov family copula, i.e.,
\begin{equation}
C_{st} = \alpha(s,t)W + (1-\alpha(s,t)-\beta(s,t))P + \beta(s,t)M,
\end{equation}
if and only if \cite{darsow1992copulas,overbeck2015multivariate}, for $s<u<t$,
\begin{eqnarray}
\alpha(s,t) &=& \beta(s,u)\alpha(u,t) + \alpha(s,u)\beta(u,t), \\
\beta(s,t) &=& \alpha(s,u)\alpha(u,t) + \beta(s,u)\beta(u,t),
\end{eqnarray}
where $\alpha(s,t)\ge{0}$, $\beta(s,t)\ge{0}$, and $\alpha(s,t)+\beta(s,t)\le{1}$. 
For homogeneous case, $\alpha(s,t)=\alpha(t-s)$ and $\beta(s,t)=\beta(t-s)$, a solution is as follows
\begin{eqnarray}
\alpha(h) &=& \frac{e^{-2h}(1 - e^{-h})}{2}, \\
\beta(h) &=& \frac{e^{-2h}(1 + e^{-h})}{2}.
\end{eqnarray}
Let $\alpha = e^{-h}$, it's a one-parameter copula \cite{darsow1992copulas}
\begin{equation}
C_\alpha = \frac{\alpha^2(1-\alpha)}{2}W + (1-\alpha^2)P + \frac{\alpha^2(1+\alpha)}{2}M,
\end{equation}
where $-1\le\alpha\le{1}$, if $|\alpha|$ is small, independence is indicated, if $\alpha$ is near $1$, strongly positive dependence is indicated, and if $\alpha$ is near $-1$, strongly negative dependence is indicated.
\end{example}

It's elaborated that Fr\'{e}chet copulas imply quite a restricted type of Markov process and Archimedean copulas are incompatible with Markov chains \cite{lageraas2010copulas}.

\section{Proof of Theorem \ref{theorem_distribution_bound}}

\begin{proof}
In order to provide exponential upper bound for the distribution of the cumulative capacity, define \cite{gallager2013stochastic}
\begin{equation}
\underline{L}(t) = \frac{\min_{j\in{E}}(h^{(\theta)}(J_j))}{h^{(\theta)}(J_0)}e^{\theta{S(t)}-t\kappa(\theta)},
\end{equation}
where $\underline{L}(t)\le L(t)$, i.e., ${E}[\underline{L}(t)]\le{1}$. 
Apply Markov inequality to $\underline{L}(t)$ and get, for any $\mu>0$, 
\begin{eqnarray}
P\{ \underline{L}(t)\ge{\mu} \} \le \frac{1}{\mu}{E}[\underline{L}(t)]\le \frac{1}{\mu}.
\end{eqnarray}
Choose $\mu=e^{-t\kappa(\theta)+\theta{x}}\frac{\min_{j\in{E}}(h^{(\theta)}(J_j))}{h^{(\theta)}(J_0)}$, for $\theta\le{0}$, 
\begin{eqnarray}
P\{ S(t)\le x \} \le \frac{h^{(\theta)}(J_0)}{\min_{j\in{E}}(h^{(\theta)}(J_j))}e^{t\kappa(\theta)-\theta{x}},
\end{eqnarray}
while for $\theta>0$,
\begin{eqnarray}
P\{ S(t)\ge x \} \le \frac{h^{(\theta)}(J_0)}{\min_{j\in{E}}(h^{(\theta)}(J_j))}e^{t\kappa(\theta)-\theta{x}},
\end{eqnarray}
which indicates that the distribution has a light tail.
Letting $-y^\ast = t\kappa(\theta)-\theta{x} \le{0}$, 
the distribution of the transient capacity is bounded by
\begin{equation}
1 - \frac{h^{(\theta)}(J_0){e^{-y_l}}}{\min\limits_{j\in{E}}(h^{(\theta)}(J_j))} \le P\left\{ \overline{C}(t) \le c^\ast \right\}\le \frac{h^{(\theta)}(J_0){e^{-y_u}}}{\min\limits_{j\in{E}}(h^{(\theta)}(J_j))},
\end{equation}
where $c^\ast=\frac{t\kappa(\theta)+y^\ast}{\theta{t}}$, with $y^\ast=y_u$ for $\theta<{0}$ for the upper bound, and $y^\ast=y_l$ for $\theta>{0}$ for the lower bound.
\end{proof}

\section{Proof of Theorem \ref{theorem_performance_bound}}

\begin{proof}
For the constant fluid arrival $A(t)=\lambda{t}$,
the delay is expressed as $P(D(t) > {x}) = P\{ \sup_{0\le{s}\le{t}}\{ \lambda(t-s) - S(s,t) \} > \lambda{x} \}$, and the relationship with $P(B(t) > {x})$ directly follows, i.e., $P(B(t) > {x}) = P\left(D(t) > {\frac{x}{\lambda}}\right)$. 

Consider the Markov additive process $S(t)-\lambda{t}$ and the likelihood ratio martingale
\begin{equation}
L(t) = \frac{h^{(\theta)}(J_t)}{h^{(\theta)}(J_0)}e^{\theta(S(t)-\lambda{t})-t\kappa(\theta)},
\end{equation}
where $\kappa(\theta)$ and $\bm{h}^{(\theta)}$ are the eigenvalue and eigenvector corresponding to $S(t)-\lambda{t}$,
with a change of measure, the delay is expressed as \cite{zhu2008ruin,asmussen2010ruin}
\begin{equation}
P_i(D\ge{x}) = {h^{(-\theta)}(J_i)} E^{(-\theta)}_i \left[ \frac{1}{h^{(-\theta)}(J_\tau)}e^{-\theta{x}} \right],
\end{equation}
where $-\theta$ is the negative root of $\kappa(\theta)=0$ of $S(t)-\lambda{t}$, $\bm{h}^{(-\theta)}$ is the corresponding right eigenvector, and $\tau$ is the stopping time that $\tau=\inf\{t\ge{0}: \lambda{t}-S(t)\ge{x}\}$.
The results follow directly.
\end{proof}

\section{Proof of Theorem \ref{theorem_no_granger_causality_biset}}

\begin{proof}
Since
\begin{IEEEeqnarray}{rCl}
\IEEEeqnarraymulticol{3}{l}{
\mathbb{P} \left( \underline{\bm{X}}_{j+1} \le \bm{x} | {\bm{X}}_j \right)  
} \nonumber\\
=  {C_{j,j+1}}_{C_{\underline{\bm{X}}_j\overline{\bm{X}}_j}, } \left( F_{\underline{\bm{X}}_j}(\underline{\bm{X}}_j), F_{\overline{\bm{X}}_j}(\overline{\bm{X}}_j), F_{\underline{\bm{X}}_{j+1}}(\bm{x}), \bm{1}_{F_{\overline{\bm{X}}_{j+1}}} \right), \IEEEeqnarraynumspace\nonumber
\end{IEEEeqnarray}
\begin{IEEEeqnarray}{rCl}
\IEEEeqnarraymulticol{3}{l}{
\mathbb{P} \left( \underline{\bm{X}}_{j+1} \le \bm{x} | {\underline{\bm{X}}}_j \right) 
} \\
=  {C_{j,j+1}}_{ C_{\underline{\bm{X}}_j}, } \left(F_{\underline{\bm{X}}_j}(\underline{\bm{X}}_j), \bm{1}_{F_{\overline{\bm{X}}_j}}, F_{\underline{\bm{X}}_{j+1}}(\bm{x}), \bm{1}_{F_{\overline{\bm{X}}_{j+1}}} \right), \IEEEeqnarraynumspace
\end{IEEEeqnarray}
the no-Granger causality holds, if and only if
\begin{multline}
{C_{j,j+1}}_{C_{\underline{\bm{X}}_j\overline{\bm{X}}_j}, } \left(\bm{u}_{\underline{\bm{X}}_j}, \bm{u}_{\overline{\bm{X}}_j}, \bm{u}_{\underline{\bm{X}}_{j+1}}, \bm{1}_{\bm{u}_{\overline{\bm{X}}_{j+1}}} \right) \\
=
{C_{j,j+1}}_{ C_{\underline{\bm{X}}_j}, } \left(\bm{u}_{\underline{\bm{X}}_j}, \bm{1}_{\bm{u}_{\overline{\bm{X}}_j}}, \bm{u}_{\underline{\bm{X}}_{j+1}}, \bm{1}_{\bm{u}_{\overline{\bm{X}}_{j+1}}} \right). 
\end{multline}
By integrating, we obtain
\begin{IEEEeqnarray}{rCl}
\IEEEeqnarraymulticol{3}{l}{
C_{j,j+1} \left(\bm{u}_{\underline{\bm{X}}_j}, \bm{u}_{\overline{\bm{X}}_j}, \bm{u}_{\underline{\bm{X}}_{j+1}}, \bm{1}_{\bm{u}_{\overline{\bm{X}}_{j+1}}} \right)
}\IEEEeqnarraynumspace\\
&=& \int_{\bm{0}}^{\bm{u}_{\underline{\bm{X}}_j}} {C_{\overline{\bm{X}}_j\underline{\bm{X}}_j}}_{, C_{\underline{\bm{X}}_j}} \left( {\bm{u}_{\overline{\bm{X}}_j}},  {\bm{u}_{\underline{\bm{X}}}} \right) \nonumber\\
&&
\cdot  {C_{\underline{\bm{X}}_j\underline{\bm{X}}_{j+1}}}_{ C_{\underline{\bm{X}}_j}, } \left(  {\bm{u}_{\underline{\bm{X}}}}, {\bm{u}_{\underline{\bm{X}}_{j+1}}} \right) C_{\underline{\bm{X}}_j}\qty(d{{\bm{u}_{\underline{\bm{X}}}}})  \IEEEeqnarraynumspace\\
&=& C_{\overline{\bm{X}}_j\underline{\bm{X}}_j} \stackrel{C_{\underline{\bm{X}}_j} (\bm{u}_{\underline{\bm{X}}_j})}{\star} C_{\underline{\bm{X}}_j\underline{\bm{X}}_{j+1}} \left(\bm{u}_{\overline{\bm{X}}_j}, \bm{u}_{\underline{\bm{X}}_{j+1}} \right). 
\end{IEEEeqnarray}
The other result follows analogically.
\end{proof}

\bibliographystyle{IEEEtran}
\bibliography{main}

\end{document}